\documentclass[reqno,12pt]{amsart}
\numberwithin{equation}{section}
\usepackage{amsfonts, amsthm}
\newtheorem{lemma}{Lemma}
\newtheorem{theorem}{Theorem}

\theoremstyle{definition}

\newcommand{\eref}[1]{\eqref{#1}}
\newcommand{\rme}{\mathrm{e}}
\newcommand{\fl}{}

\selectfont 

\usepackage{hyperref}
\usepackage{orcidlink}

\begin{document}
\title[New family of Darboux integrable equations]{A new family of Darboux integrable partial differential equations}

\author{S. Ya. Startsev\,\,\orcidlink{0000-0001-5891-6191}}\thanks{\href{http://www.researcherid.com/rid/D-1158-2009}{Web of Science ResearcherID: D-1158-2009}}
\address{Institute of Mathematics, Ufa Federal Research Centre, Russian Academy of Sciences}  

\begin{abstract}
A new family of Darboux integrable partial differential equations is constructed. The minimal orders of integrals in both characteristics can simultaneously be arbitrarily high for equations within this family. This suggests that the complete list of Darboux integrable equations may contain substantially more equations than are currently known.
\end{abstract}

\vspace{2pc}

\maketitle

\section{Introduction: the main result with an outline of the proof}\label{intro}
One of the classical open problems (see, for example, \cite{Gurs}) is to find all partial differential equations
\begin{equation}\label{hyp}
u_{xy}=F(x,y,u,u_x,u_y)
\end{equation}
for which there exist functions  
\begin{eqnarray}
w(x,y,u,u_1, \dots, u_k), \qquad&w_{u_k} \ne 0,\quad k>0, \qquad&u_i:=\partial^i u /\partial x^i, \label{xi}\\
\bar{w}(x,y,u,\bar{u}_1, \dots, \bar{u}_m), \qquad&\bar{w}_{\bar{u}_m}\ne 0,\quad m>0, \qquad&\bar{u}_i:=\partial^i u /\partial y^i,
\end{eqnarray} 
such that for any solution of~\eref{hyp} the function $w$  depends only on $x$, and the function $\bar{w}$ depends only on $y$. Such equations are said to be \emph{Darboux integrable}, and the functions $w$ and $\bar{w}$ are respectively called an \emph{$x$-integral of order $k$} and a \emph{$y$-integral of order $m$} for the equation~\eref{hyp}. An integral of the lowest  order for an equation is called a \emph{minimal integral} of this equation. 

The best-known examples of Darboux integrable equations are the wave equation $u_{xy}=0$ with $w=u_x$ and $\bar{w}=u_y$ and the Liouville equation $u_{xy}=\rme^u$ admitting the integrals 
\[ w= u_{xx} - u_x^2/2,  \qquad \bar{w}= u_{yy} - u_y^2/2. \] 
It is easy to check these (and other) examples by taking into account that $w$ and $\bar{w}$ are $x$- and $y$-integrals iff $D_y(w)=0$ and $D_x(\bar{w})=0$, where 
\begin{eqnarray}
&D_{x} &= \frac{\partial}{\partial x}+ u_1 \frac{\partial}{\partial u}  + \sum^{\infty
}_{i=1} \left( u_{i+1} \frac{\partial} {\partial u_i} 
+ D^{i-1}_{y}(F) \frac{\partial}{\partial
\bar{u}_i} \right), \label{dx} \\
&D_{y} &=
\frac{\partial}{\partial y}+ \bar{u}_1 \frac{\partial}{\partial
u}  +\sum^{\infty }_{i=1}\left( \bar{u}_{i+1} \frac{\partial}{\partial \bar{u}_i}  + D^{i-1}_x (F) \frac{\partial}{\partial
u_i}  \right) \label{dy}
\end{eqnarray}
are the total derivatives with respect to $x$ and $y$ by virtue of the equation~\eref{hyp}. (In the last two formulas and further we use the standard convention that the zeroth power of any operator is the identity map; so $D_x^0(F)=D_y^0(F)=F$ and the like.) 

Despite the long history of the study, interest in the Darboux integrable equations~\eref{hyp} continues to this day -- see, for example, relatively recent works \cite{ZhIzv}--\cite{St25}. In particular, a new family of Darboux integrable equations was found in~\cite{St25}. A part of this family consists of the equations
\begin{equation}\label{seq}
 u_{xy}= D_y \left( \sum_{i=1}^{n} \xi_i(y) \eta_i(x) \rme^{-u} \right), 
\end{equation}
which are Darboux integrable for any functions $\xi_i$, $\eta_i$. The integrals of~\eref{seq} are 
\[w=u_x - \sum_{i=1}^{n} \xi_i(y) \eta_i(x) \rme^{-u}, \quad \bar{w} = D_y (\ln \Lambda (\rme^u)), \]
where $\Lambda$ is an operator $\sum_{j=0}^{n} \lambda_j (y) D_y^j$ annihilating all $\xi_i$. 
The $x$-integral is always minimal due to its first order, and the $y$-integral is minimal if the two Wronskians are non-zero:
\begin{equation}\label{wron}
\left| \begin{array}{cccc}
\xi_1 & \xi_2 & \ldots & \xi_n\\
\xi'_1 & \xi'_2 & \ldots & \xi'_n\\
\vdots& \vdots & \ddots & \vdots\\
\xi_1^{(n-1)} & \xi_2^{(n-1)} & \ldots & \xi_n^{(n-1)}
\end{array} \right| 
\left| \begin{array}{cccc}
\eta_1 &  \eta'_1 & \ldots & \eta_1^{(n-1)}\\
\eta_2 &  \eta'_2 & \ldots & \eta_2^{(n-1)}\\
\vdots & \vdots & \ddots & \vdots\\
\eta_n &  \eta'_n & \ldots & \eta_n^{(n-1)}
\end{array} \right| \ne 0.
\end{equation}

Along with this, in \cite{SvSok} differential substitutions of a special class were considered and it was noted that such a substitution relates the generic equation 
\begin{equation}\label{exe}
u_{xy}=D_x(\alpha(x,y) \rme^u)+D_y(\beta(x,y) \rme^{-u}),  
\qquad |\alpha|+|\beta| \ne 0,
\end{equation}
to an equation of the same form. Thus, we can apply this differential substitution many times (infinitely many times in the generic case and $n$ times if we start from~\eref{seq}). As it was proved in \cite{Sljm}, differential substitutions of this class preserve the Darboux integrability. All this allows us to generate a sequence of Darboux integrable equations by starting from~\eref{seq}.

The main goal of the present work is to explicitly construct this sequence. The corresponding result is formulated in the following statement. 
\begin{theorem}\label{main}
Let us denote $\alpha_0=0$, $\beta_0=\Delta_0=\beta(x,y)$,
\begin{equation}\label{delt}
\Delta_i = 
\left| \begin{array}{cccc}
\beta & \beta_x & \ldots & \frac{\partial^i \beta}{\partial x^i}\\
\beta_y & \beta_{xy} & \ldots & \frac{\partial^{i+1} \beta}{\partial x^i \partial y}\\
\vdots& \vdots & \ddots & \vdots\\
\frac{\partial^i \beta}{\partial y^i} & \frac{\partial^{i+1} \beta}{\partial x \partial y^i} & \ldots & \frac{\partial^{2i} \beta}{\partial x^i \partial y^i}
\end{array} \right| , \quad \beta_i = \frac{\Delta_i}{\Delta_{i-1}}, \quad \alpha_i=\frac{1}{\beta_{i-1}}, \quad i>0.
\end{equation}

If $\beta= \sum_{i=1}^{n} \xi_i(y) \eta_i(x)$ and $\Delta_{n-1} \ne 0$, then $\Delta_i \ne 0$ for all $i < n$ and the equations
\begin{equation}\label{abe}
u_{xy}=D_x(\alpha_i \rme^u)+D_y(\beta_i \rme^{-u}), \qquad i=\overline{0,n}, 
\end{equation}
are Darboux integrable. The minimal $x$- and $y$-integrals of the equation~\eref{abe} are 
\begin{equation}\label{abei}
w=D_x(\ln \Upsilon_i(\rme^{-u})) + \beta_i \rme^{-u}, \qquad \bar{w}=D_y(\ln \Lambda_i(\rme^u)) - \alpha_i \rme^u,
\end{equation}
where $\Upsilon_0$, $\Lambda_n$ are the identity map, and
 \[ \fl \qquad \Upsilon_i = \left( D_x + \beta_i \rme^{-u} - (\ln\alpha_1)_x \right) \circ \dots \circ \left( D_x + \beta_i \rme^{-u} - (\ln\alpha_i)_x \right), \qquad i>0, \qquad \]
\begin{equation}\label{lami} 
\fl \qquad \Lambda_i = \left( D_y - \alpha_i \rme^u - (\ln\beta_{n-1})_y \right) \circ \dots \circ \left( D_y - \alpha_i \rme^u - (\ln\beta_{i})_y \right), \qquad i<n, 
\end{equation}
are the differential operators of order $i$ and $n-i$, respectively. Here the symbol $\circ$ denotes the composition of operators.
\end{theorem}
Along the way, we also prove that the equations~\eref{abe} exhaust all Darboux integrable equations of the form~\eref{exe} (see Lemma~\ref{trall} below). 

In the generic case, the equations~\eref{abe} are very likely new because they cannot be related via point transformations to any other Darboux integrable equations known to the author. Accordingly, Theorem~\ref{main} indicates that new Darboux integrable equations~\eref{hyp} may be found not only among those possessing first-order integrals in one of the characteristics (as it was demonstrated in~\cite{St25}), but also in the domain of equations with minimal integrals of order 2 and higher in both characteristics. This suggests that the complete list of Darboux integrable equations may be substantially larger than what is currently known.
 
\section{Equations with first-order integrals in one of the characteristics} 
It is easy to see that the equation~\eref{seq} coincides with the equation~\eref{abe} when $i=0$, and the condition $\Delta_{n-1} \ne 0$ coincides with \eref{wron}. In this section we show that the operator $\Lambda_0$ can serve as an operator $\Lambda$ defining the $y$-integral of~\eref{seq} (i.e., that $\Lambda_0$ annihilates all $\xi_i$ and its coefficients do not depend on $x$). A part of the reasoning (e.g., Lemmas~\ref{l1} and \ref{l2}) is not necessary for this purpose, but allows us to later prove that the equations~\eref{abe} are the only Darboux integrable equations in the class~\eref{exe}. 
Accordingly, we initially make no assumptions on $\beta$.

\begin{lemma}\label{l1}
An equation of the form~\eref{exe} admits a first-order $x$-integral if and only if $\alpha=0$.
\end{lemma}
\begin{proof} We can rewrite any function $w(x,u,u,u_x)$ in the form $\theta(x,y,u,\varphi)$, where $\varphi=u_x-\beta \rme^{-u}$. Equation~\eref{exe} implies $D_y(\varphi)=\psi$, where $\psi = D_x(\alpha \rme^u)$. If $\theta$ is an $x$-integral of~\eref{exe}, then $D_y(\theta) = \theta_y + u_y \theta_u +\psi \theta_\varphi = 0$. This equality means that $\theta_u=0$ and $\psi$ is a function of $x$, $y$, $\varphi$. The latter implies
\[ \varphi_{u_x} \psi_u - \varphi_u \psi_{u_x} = \alpha_x \rme^u + \alpha \rme^u u_x - \alpha \beta = 0, \] 
but this can be true only if $\alpha=0$. The equation~\eref{exe} obviously 
has the $x$-integral $w=u_x-\beta \rme^{-u}$ if $\alpha=0$. 
\end{proof}

\begin{lemma}\label{l2}
An equation of the form 
\begin{equation}\label{be}
u_{xy}=D_y(\beta(x,y) \rme^{-u}), \qquad \beta \ne 0,
\end{equation}
admits a $y$-integral only if the linear equation 
\begin{equation}\label{le}
v_{xy} = \frac{\beta_y}{\beta}  v_x
\end{equation}
admits a $y$-integral.
\end{lemma}
In fact, the above lemma is a highly specific case of a more  general statement (see Theorem~2 in~\cite{Sljm}). For the reader's convenience, we give an independent proof of the lemma. 
\begin{proof} Let $\bar{w}$ be a $y$-integral of order $m$ for~\eref{be}. Then 
\[ D_x(\bar{w}) = \bar{w}_x + u_x \bar{w}_u  + \sum^{m}_{i=1} D^{i}_{y}(\beta(x,y) \rme^{-u}) \bar{w}_{\bar{u}_i} =0. \]
Differentiating the last equality with respect to $u_x$, we obtain $\bar{w}_u=0$. As noted in~\cite{SvSok}, the differential 
substitution $u=\ln(\beta v_x /v)$ maps solutions of the linear equation $v_{xy} = \beta_y v_x / \beta - \alpha \beta v$ into solutions of~\eref{exe}. The $y$-integral is a function only of $y$ for any solution of~\eref{be} and, in particular, for all solutions $u=\ln(\beta v_x /v)$ obtained from those of~\eref{le}. This is why 
\[ \bar{\theta} = \bar{w}(x,y,D_y(\ln(\beta v_x /v)),\dots,D_y^m(\ln(\beta v_x /v))) \]
is a function only of $y$ for any solution of~\eref{le}. To prove that $\bar{\theta}$ is a $y$-integral of~\eref{le}, it remains to ensure that $\bar{\theta}$ depends on derivatives of $v$ with respect to $y$. This follows from the fact that $D_y(\ln(\beta v_x /v))= 2\beta_y /\beta - v_y/v$ on solutions of~\eref{le}.  
\end{proof}

Laplace transformations (see, for example, \cite{Tr}) can be used to construct the integrals of linear equations $L(u)=0$, where
\[ L= D_y D_y - A(x,y) D_x - B(x,y) D_y - C(x,y). \] 
This is done as follows.  We rewrite the operator $L$ as
\begin{equation}\label{delim}
L= (D_x - B) \circ (D_y - A) - H_1, \qquad H_1 = C + A B - A_x.
\end{equation}
Starting with $L_0=L$, we construct the sequence of operators
\begin{equation}\label{lid}
L_i := (D_y - A_i) \circ (D_x - B)  - H_i = (D_x - B) \circ (D_y - A_i) - H_{i+1},
\end{equation}
where $A_0=A$, $A_i= A_{i-1} + (\ln H_i)_y$, $H_{i+1}=H_i + B_y - (A_i)_x$. 
A direct calculation shows that
\begin{equation}\label{cint}
(D_y - A_i) \circ L_{i-1} = L_i \circ (D_y - A_{i-1}). 
\end{equation}
It is known (see, for example, \cite{AK}) that $H_n=0$ for some $n \le p$ if $L(u)=0$ has a $y$-integral of order $p$. Conversely, if $H_n=0$ and $n>1$, then~\eref{cint} and~\eref{lid} imply
\begin{eqnarray*}
\fl \qquad \qquad (D_y - A_{n-1}) \circ \dots \circ (D_y - A_1) \circ L &= L_{n-1} \circ (D_y - A_{n-2}) \circ \dots \circ (D_y - A_0) \\
\fl = (D_x - B) \circ (D_y - A_{n-1}) \circ \dots \circ (D_y - A_0) &= \mu^{-1} D_x \circ \mu (D_y - A_{n-1}) \circ \dots \circ (D_y - A_0),
\end{eqnarray*}
where $\mu_x = - B \mu$. Taken together, the above facts (including~\eref{delim} in the case $n=1$) imply that \mbox{$\mu (D_y - A_{n-1}) \dots (D_y - A_0)(u)$} is a minimal $y$-integral of the linear equation. Applying this procedure to~\eref{le}, we obtain the unnumbered part and item~1) of the following statement.

\begin{lemma}\label{lint}
If equation~\eref{le} admits a $y$-integrals, then $H_n=0$ for some $n>0$, where 
\begin{equation}\label{hi}
H_1 = - (\ln \beta)_{xy}, \qquad H_{i+1} = H_i - \left( \ln (\beta H_1 \dots H_i) \right)_{xy}, \quad i>0.
\end{equation}

Conversely, if $H_n=0$, then:

 1) $\Lambda_0 (v)$ is a minimal $y$-integral for~\eref{le}, where 
\begin{equation}\label{lamb}
\fl \Lambda_0=(D_y - A_{n-1}) \circ \dots \circ (D_y - A_0), \qquad A_0=(\ln \beta)_y, \quad A_i=\left( \ln (\beta H_1 \dots H_i) \right)_y, \quad i>0; 
\end{equation}

2) When $\Lambda_0$ is rewritten in the form $\sum_{i=0}^{n} \lambda_i D_y^i$,  its coefficients $\lambda_i$ depend only on $y$;

3) The function $\beta$ has the form $\sum_{i=1}^{n} \xi_i(y) \eta_i(x)$, where $\Lambda_0(\xi_i)=0$.
\end{lemma}
\begin{proof} 
In view of the reasoning just before the lemma, it remains to prove items~2) and 3).

Let $\bar{v}_i$ denote $\partial^i v /\partial y^i$. Since $\Lambda_0(v)= \sum_{i=0}^n \lambda_i(x,y) \bar{v}_i$ is a $y$-integral of~\eref{le}, we have
\[ D_x(\Lambda_0(v)) = \sum_{i=0}^n (\lambda_i)_x \bar{v}_i + \left(\lambda_0 + \sum_{i=1}^n \lambda_i \left( D_y + \frac{\beta_y}{\beta} \right)^{i-1} \left( \frac{\beta_y}{\beta} \right) \right) v_x = 0.\]
Hence, $(\lambda_i)_x=0$ for all $i=\overline{0,n}$.

Because $(D_y-A_0)(\beta)=(D_y-\beta_y/\beta)(\beta)=0$, the function $\beta(x,y)$ belongs to the kernel of $\Lambda_0$. Therefore, $\beta=\sum_{i=1}^{n} \xi_i(y) \eta_i(x)$, where $\xi_i$ are $n$ linearly independent solutions of the equation $\Lambda_0(\xi)=0$.
\end{proof}

For further reasoning we employ the relation (see \cite{Darb}, p.~129)
\begin{equation}\label{darf}
(\ln \Delta_j)_{xy} = \frac{\Delta_{j-1} \Delta_{j+1}}{\left(\Delta_j\right)^2}, \qquad j>0, 
\end{equation}
where $\Delta_0=\beta$ and the remaining $\Delta_i$ are defined by \eref{delt}. This formula can be easily derived by using the Desnanot-Jacoby identity (a particular case of the Sylvester's determinant identity)
\[ M_{j+1}^{j+1} M_{j+2}^{j+2} - M_{j+2}^{j+1} M_{j+1}^{j+2} = M_{j+1,j+2}^{j+1,j+2} M , \]
where $M$ is the determinant of a $(j+2) \times (j+2)$ matrix, $M_p^q$ is the minor of the same matrix (the determinant of the matrix after deleting row $p$ and column $q$) and $M_{j+1,j+2}^{j+1,j+2}$ is the $j$-th order leading principal minor of the matrix. Indeed, if $M=\Delta_{j+1}$, then
\begin{eqnarray*}
&(\Delta_j)^2 (\ln \Delta_j)_{xy} &= (\Delta_j)_{xy} \Delta_j - (\Delta_j)_x (\Delta_j)_y \\
&&= M_{j+1}^{j+1} M_{j+2}^{j+2} - M_{j+2}^{j+1} M_{j+1}^{j+2} = M_{j+1,j+2}^{j+1,j+2} M  = \Delta_{j-1} \Delta_{j+1}. 
\end{eqnarray*}

\begin{lemma}\label{hid} Let $H_i$ and $\Delta_i$ be defined by \eref{hi} and \eref{delt}, respectively, and $\Delta_0=\beta$. Then
\begin{equation}\label{hd}
H_j= - (\ln \Delta_{j-1})_{xy},  \qquad \beta H_1 \dots H_j = (-1)^j \frac{\Delta_j}{\Delta_{j-1}}
\end{equation}
for all $j$ such that $\Delta_0 \dots \Delta_{j-1} \ne 0$.
\end{lemma}
\begin{proof} 
Obviously, the equalities~\eref{hd} hold for $j=1$. Using~\eref{darf} and induction on $j$, we obtain
\[ \fl H_{j+1} = H_j - \left( \ln (\beta H_1 \dots H_j) \right)_{xy} = - (\ln \Delta_{j-1})_{xy} - \left(\ln \left(\frac{\Delta_j}{\Delta_{j-1}}\right) \right)_{xy} = - (\ln \Delta_j)_{xy}, \]
\[ \fl \beta H_1 \dots H_{j+1} = (-1)^{j+1} \frac{\Delta_j}{\Delta_{j-1}} (\ln \Delta_j)_{xy} = (-1)^{j+1} \frac{\Delta_j}{\Delta_{j-1}} \frac{\Delta_{j-1} \Delta_{j+1}}{\left(\Delta_j\right)^2} = (-1)^{j+1} \frac{\Delta_{j+1}}{\Delta_j}. \qedhere \]
\end{proof}
\begin{theorem}\label{foi}
An equation of the form~\eref{exe} admits a $y$-integral and a first-order $x$-integral if and only if $\alpha =0$ and for some $n>0$ the function $\beta$ can be represented in the form $\sum_{i=1}^{n} \xi_i(y) \eta_i(x)$ such that $\Delta_{n-1} \ne 0$, where the determinants $\Delta_i$ are defined by \eref{delt}. The above representation of $\beta$ guarantees that $\Delta_i \ne 0$ for all $i \le n-1$ and that the minimal integrals of the equation are
\[w=u_x - \sum_{i=1}^{n} \xi_i(y) \eta_i(x) \rme^{-u}, \quad \bar{w} = D_y (\ln \Lambda_0 (\rme^u)), \]
where $\Lambda_0$ is defined by~\eref{lami} with $i=0$. 
\end{theorem}
\begin{proof} According to Lemma~\ref{l1}, the equation admits a first-order $x$-intergral iff $\alpha=0$. If~\eref{exe} has a $y$-integral, then $H_n=0$ and $\beta=\sum_{i=1}^{n} \xi_i(y) \eta_i(x)$ by Lemmas~\ref{l2} and~\ref{lint}. Lemma~\ref{hid} implies that $n$ is the smallest integer for which $\Delta_n=0$. Therefore, $\Delta_{n-1} \ne 0$.

Conversely, if  $\beta=\sum_{i=1}^{n} \xi_i(y) \eta_i(x)$, then $\Delta_n$ is the determinant of a matrix product (like 
the left-hand side of~\eref{wron}) and vanishes by the Cauchy-Binet theorem. 
Therefore, either $H_n=0$ or $\Delta_0 \dots \Delta_{n-1} = 0$ by Lemma~\ref{hid}. Thus, $\Delta_p=0$ and, hence, $H_p=0$ for some $p \le n$. Lemma~\ref{lint} then implies that the coefficients $\lambda_i$ of the operator
$\Lambda_0=(D_y - A_{p-1}) \circ \dots \circ (D_y - A_0)$
do not depend on $x$.  In addition, $\sum_{i=1}^{n} \eta_i \Lambda_0(\xi_i) =\Lambda_0(\beta)=0$ by construction of $\Lambda_0$. Differentiating the last equality $n-1$ times with respect to $x$, we obtain a linear homogeneous system of $n$ equations for $\Lambda_0(\xi_i)$. Since $\Delta_{n-1} \ne 0$ (i.e. the condition~\eref{wron} holds), the matrix of this system is non-degenerate and  $\Lambda_0(\xi_i)=0$ for all $i=\overline{1,n}$. This means that $p=n$ because $\xi_i$ are linearly independent due to~\eref{wron}. In view of~\eref{hd}, formula~\eref{lamb} coincides with~\eref{lami} for $i=0$.

Let us demonstrate that $D_y(\ln \Lambda_0(\rme^u))$ is a $y$-integral of~\eref{seq}. Taking $D_y(w)=0$ and $u_x= w+\sum_{i=1}^{n} \xi_i(y) \eta_i(x) \rme^{-u}$ into account, we obtain
\[ D_x(\Lambda_0(\rme^u))= \Lambda_0(\rme^u u_x) = \Lambda_0( w \rme^u) + \sum_{i=1}^{n} \eta_i \Lambda_0(\xi_i)  = w \Lambda_0(\rme^u),\]
\begin{eqnarray*}
D_x\left(\frac{D_y(\Lambda_0(\rme^u))}{\Lambda_0(\rme^u)} \right) &= \displaystyle \frac{D_y(D_x(\Lambda_0(\rme^u)))}{\Lambda_0(\rme^u)} - \frac{D_y(\Lambda_0(\rme^u)) D_x(\Lambda_0(\rme^u))}{(\Lambda_0(\rme^u))^2}  \\ 
&= \displaystyle \frac{D_y( w \Lambda_0(\rme^u))}{\Lambda_0(\rme^u)} - \frac{D_y(\Lambda_0(\rme^u)) w \Lambda_0(\rme^u)}{(\Lambda_0(\rme^u))^2} = 0. 
\end{eqnarray*}
It was proved in Example~1 of \cite{St25} that the equation~\eref{seq} has no $y$-integrals of order less than $n+1$ under the condition~\eref{wron} (i.e., if $\Delta_{n-1} \ne 0$).
\end{proof}

\section{Transformations preserving Darboux integrability}
Note that the equation~\eref{exe} can be represented\footnote{If for equation~\eref{hyp} there exist $\varphi$ and $\psi$ satisfying~\eref{cl}, then the right-hand side $F$ of the equation is uniquely defined by~\eref{cl} and we can consider~\eref{cl} as a form of writing for~\eref{hyp}.} in the form
\begin{equation}\label{cl}
D_x ( \varphi (x, y, u, u_y) ) = \psi (x,y,u,u_y), \qquad \varphi _{u_y} \psi _u - \varphi _{u} \psi_{u_y} \ne 0.
\end{equation}
Due to this, we can apply a transformation proposed in \cite{Yam,SvSok} to equation~\eref{exe}. This transformation looks as follows.

The relation~\eref{cl} is equivalent to the system
\begin{equation}\label{it}
v=\varphi (x, y, u, u_y), \qquad  v_x = \psi (x,y,u,u_y).
\end{equation}
Solving \eref{it} for $u$, $u_y$, we arrive at a system of the form
\begin{equation}\label{pqs}
u=\varrho (x,y,v,v_x), \qquad u_y = \sigma (x, y, v, v_x).
\end{equation}
Differentiating the first equation of~\eref{pqs} with respect to
$y$ and then comparing the result with the second equation, we obtain
\begin{equation}\label{pq}
D_y ( \varrho (x,y,v,v_x) ) = \sigma (x, y, v, v_x).
\end{equation}
As a result, the differential substitution $v=\varphi (x, y, u, u_y)$ maps solutions of~\eref{cl} into solutions of~\eref{pq}. Repeating the above reasoning in reverse order, we see that the differential substitution $u=\varrho (x,y,v,v_x)$ maps solutions of~\eref{pq} into solutions of~\eref{cl}. 

\begin{theorem}[\cite{Sljm}]\label{ljmt}
Let equation~\eref{hyp} be written  in the form~\eref{cl}. Then~\eref{hyp} admits an $x$-integral $w$ of order $k$ and a $y$-integral $\bar{w}$ of order $m$ if and only if the corresponding equation~\eref{pq} admits an $x$-integral of order $k+1$ and a $y$-integral of order $m-1$. Specifically, the $y$-integral $\bar{w}$ can be written in the form $\bar{\theta} (x,y, \varphi, D_y(\varphi), \dots, D_y^{m-1} (\varphi))$, and the equation~\eref{pq} admits the $y$-integral $\bar{\theta}(x,y,v,\dots,\bar{v}_{m-1})$ and the $x$-integral
\begin{equation}\label{pqy}
\theta = w(x,y,\varrho,D_x(\varrho),\dots, D_x^k(\varrho))
\end{equation}
of orders $m-1$ and $k+1$, respectively. The functions $\theta$ and $\bar{\theta}$ are minimal integrals for~\eref{pq} if $w$ and $\bar{w}$ are minimal integrals for~\eref{hyp}. 
\end{theorem}

Applying the above scheme and Theorem~\ref{ljmt} to \eref{exe}, we obtain the following statement. 
\begin{lemma}\label{trall}
If $\beta \ne 0$, then the differential substitution $v=\ln (\beta_y + \alpha \beta \rme^u - \beta u_y)$ maps solutions of \eref{exe} into solutions of the equation
\begin{equation}\label{texe}
v_{xy} = D_x (\tilde{\alpha} \rme^v) + D_y (\tilde{\beta} \rme^{-v}),
\end{equation}
where $\tilde{\alpha}=\beta^{-1}$, $\tilde{\beta}=\beta ((\ln \beta)_{xy} + \alpha \beta)$. All Darboux integrable equations of the form~\eref{exe} either coincide with the equations~\eref{seq} satisfying the condition \eref{wron}, or are obtained from them by applying the above substitution one or multiple times. 
\end{lemma}

\begin{proof} 
Instead of the obvious representation $D_x(u_y - \alpha \rme^u) = D_y (\beta \rme^{-u})$, we write~\eref{exe} as 
\begin{equation}\label{rep}
D_x(\ln (\beta_y + \alpha \beta \rme^u - \beta u_y)) = \frac{\beta_{xy} + \beta_x (\alpha \rme^u - u_y) - \beta D_y (  \beta \rme^{-u})}{\beta_y + \alpha \beta \rme^u - \beta u_y}.
\end{equation}
(I.e., we consider the composition of the substitution $\tilde{v}=u_y - \alpha \rme^u$ and the change $v=\ln (\beta_y - \beta \tilde{v})$ to again obtain an equation of the form~\eref{exe}.)  
Then the corresponding system~\eref{it} is
\begin{equation*}
\rme^v=\beta_y + \alpha \beta \rme^u - \beta u_y,  \qquad \rme^v v_x = \beta_{xy} + \beta_x (\alpha \rme^u - u_y) + \beta \rme^{-u} (\beta u_y - \beta_y).
\end{equation*}
Solving its first equation for $u_y$ and substituting the result into the second one, we have
\begin{equation}\label{uy}
u_y=\frac{\beta_y + \alpha \beta \rme^u - \rme^v}{\beta}, 
\end{equation}
\[ \rme^v v_x  = \beta_{xy} + \beta_x \frac{\rme^v-\beta_y}{\beta} + \beta \rme^{-u} (\alpha \beta \rme^u - \rme^v) = \beta_{xy} + \frac{\beta_x}{\beta} \rme^v - \frac{\beta_x \beta_y}{\beta} + \alpha \beta^2 - \beta \rme^{v-u}. \]
The last equality gives us the expression
\begin{equation}\label{emu}
\rme^{-u} = \frac{\beta_x}{\beta^2} - \frac{v_x}{\beta} + \left(\frac{\beta_{xy}}{\beta} -\frac{\beta_x \beta_y}{\beta^2} + \alpha \beta \right) \rme^{-v} = - \tilde{\alpha}_{x} - \tilde{\alpha} v_x + \tilde{\alpha} \tilde{\beta} \rme^{-v}.
\end{equation}
Substituting~\eref{emu} into both sides of~\eref{uy}, we see that the corresponding equation~\eref{pq} is 
\begin{equation}\label{pqf}
\frac{D_y(\tilde{\alpha} \tilde{\beta} \rme^{-v} - \tilde{\alpha}_{x} - \tilde{\alpha} v_x)}{\tilde{\alpha} \tilde{\beta} \rme^{-v} - \tilde{\alpha}_{x} - \tilde{\alpha} v_x} = \frac{(\tilde{\alpha} \rme^v - \beta_y /\beta) (\tilde{\alpha} \tilde{\beta} \rme^{-v} - \tilde{\alpha}_{x} - \tilde{\alpha} v_x) - \alpha}{\tilde{\alpha} \tilde{\beta} \rme^{-v} - \tilde{\alpha}_{x} - \tilde{\alpha} v_x}.
\end{equation}
Taking $\tilde{\alpha} \beta_y /\beta = - \tilde{\alpha}_y$ into account and comparing the numerators on the left-hand and right-hand sides of~\eref{pqf}, we obtain
\[ \fl \tilde{\alpha} D_y(\tilde{\beta} \rme^{-v}) + \tilde{\alpha}_y (\tilde{\beta} \rme^{-v} - v_x) - \tilde{\alpha}_{xy} - \tilde{\alpha} v_{xy} 
= - \tilde{\alpha} D_x (\tilde{\alpha} \rme^v) + \tilde{\alpha}_y (\tilde{\beta} \rme^{-v} - v_x) + \tilde{\alpha}^2 \tilde{\beta} + \tilde{\alpha}_x \frac{\beta_y}{\beta} - \alpha \]
It is easy to check that $\tilde{\alpha}^2 \tilde{\beta} + \tilde{\alpha}_x \beta_y /\beta - \alpha =  - \tilde{\alpha}_{xy}$ and, therefore, \eref{pqf} is equivalent to the equation~\eref{texe}.

Reversing the above reasoning, we see that the substitution $u=-\ln(\tilde{\alpha} \tilde{\beta} \rme^{-v} - \tilde{\alpha}_{x} - \tilde{\alpha} v_x) $ maps solutions of~\eref{texe} into solutions of~\eref{exe} if $\tilde{\alpha} \ne 0$ (otherwise, \eref{texe} has the first order $x$-integral $v_x-\tilde{\beta} \rme^{-v}$). We can repeat this substitution until we arrive at an equation of the form~\eref{exe} with $\alpha=0$. By Theorems~\ref{ljmt} and \ref{foi}, the later will happen after $k$ applications of the substitution if~\eref{texe} is Darboux integrable with a minimal $x$-integral of order $k+1$. The final equation has the form~\eref{seq} with $\Delta_{n-1}\ne 0$ by Theorem~\ref{foi}. All Darboux integrable equations~\eref{exe} are obtained from this equation via the substitutions described in the Lemma.
\end{proof}

\begin{proof}[Proof of Theorem~\ref{main}] 
If we start from the equation~\eref{seq}, then after $i$-fold application of the transformation from Lemma~\ref{trall} we arrive at the equation
\begin{equation}\label{aben}
u_{xy} = D_x (\alpha_i \rme^u) + D_y (\beta_i \rme^{-u}).
\end{equation}
The equations~\eref{aben} are Darboux integrable by Theorem~\ref{ljmt}. Lemma~\ref{trall} implies that
\begin{equation}\label{abi}
\alpha_{i+1} =1/\beta_i, \quad \beta_{i+1}=\beta_i ((\ln \beta_i)_{xy} + \alpha_i \beta_i), \quad \alpha_0 = 0, \quad \beta_0 = \sum_{i=1}^{n} \xi_i(y) \eta_i(x).
\end{equation}

Note that $\beta_1=-\beta_0 H_1$, where $H_i$ are defined by~\eref{hi} with $\beta=\beta_0$. An induction on $i$ shows  that 
$\beta_i = (-1)^i \beta_0 H_1 \dots H_i$ for all $i>1$ with $H_1 \dots H_{i-1} \ne 0$.
Indeed, substituting the last equality into the second formula of \eref{abi}, taking $\alpha_i \beta_i = \beta_i/\beta_{i-1} = - H_i$ into account and comparing the result with~\eref{hi}, we obtain
\[ \beta_{i+1} = (-1)^i \beta_0 H_1 \dots H_i \left( (\ln(\beta_0 H_1 \dots H_i))_{xy} - H_i \right) = (-1)^{i+1} \beta_0 H_1 \dots H_i H_{i+1}. \]
The application of Lemma~\ref{hid} to the above formula for $\beta_i$ yields formulas~\eref{delt} for the coefficients of~\eref{aben}. Theorem~\ref{foi} ensures the correctness of dividing by $\beta_i$ and $H_i$ as it guarantees that $\Delta_i \ne 0$ and, hence, $\beta_i H_i \ne 0$ for all $i < n$. Theorem~\ref{foi} also implies Theorem~\ref{main} for $i=0$.

Let us apply Theorem~\ref{ljmt}, assuming that the minimal integrals of~\eref{aben} are defined by formulas~\eref{abei} for some $i<n$. The proof of Lemma~\ref{trall} (see \eref{rep}, \eref{emu}) shows that
\begin{equation}\label{fb}
\fl \varphi = \ln ((\beta_i)_y + \alpha_i \beta_i \rme^u - \beta_i u_y), \qquad \varrho = - \ln ( \alpha_{i+1} \beta_{i+1} \rme^{-v} - (\alpha_{i+1})_{x} -\alpha_{i+1} v_x). 
\end{equation} 
As it is easy to check, $( D_y - \alpha_i \rme^{u} - (\ln\beta_{i})_y)(\rme^u) = - \rme^{\varphi} \rme^u/\beta_i$ and
\begin{equation}\label{opr}
\left( D_y - \alpha_i \rme^u - (\ln\beta_{j})_y \right) \circ \frac{\rme^u}{\beta_i} = \frac{\rme^u}{\beta_i} \left( D_y - \alpha_{i+1} \rme^{\varphi} - (\ln\beta_{j})_y \right).
\end{equation}
The coefficient of $\rme^{\varphi}$ in~\eref{opr} is obtained by taking $\alpha_{i+1} = 1/\beta_i$ into account. The last three equalities allow us to rewrite the minimal $y$-integral $\bar{w}=D_y(\ln\Lambda_i(\rme^u)) - \alpha_i \rme^u$ of~\eref{aben} as
\begin{equation}\label{abin}
\bar{w} = D_y\left(\ln\left(\frac{\rme^u}{\beta_i} \Lambda_{i+1}(\rme^\varphi)\right)\right) - \alpha_i \rme^u = D_y(\ln\Lambda_{i+1}(\rme^{\varphi})) - \alpha_{i+1} \rme^{\varphi}, 
\end{equation}
where $\Lambda_{i+1}$ is the identity map for $i+1=n$, and is defined by \eref{lami} with $u=\varphi$ for $i+1 < n$.
By Theorem~\ref{ljmt},  the rightmost side of~\eref{abin} after replacing $\varphi$ with $v$ yields the minimal $y$-integral for the equation
\begin{equation}\label{abenp}
v_{xy} = D_x (\alpha_{i+1} \rme^v) + D_y (\beta_{i+1} \rme^{-v}).
\end{equation}
Thus, the second formula for the minimal integrals in~\eref{abei} is proved by induction on $i$. 

The second equation in~\eref{fb} can be written as
\begin{equation}\label{emp}
\rme^{-\varrho} = \alpha_{i+1} (\beta_{i+1} \rme^{-v} - (\ln\alpha_{i+1})_x - v_x).
\end{equation}
This equality and $\beta_i \alpha_{i+1} = 1$ imply
\[ \left(D_x + \beta_i \rme^{-\varrho} - (\ln \alpha_j)_x \right) \circ \alpha_{i+1} \rme^v = \alpha_{i+1} \rme^v \left(D_x + \beta_{i+1} \rme^{-v} - (\ln \alpha_j)_x \right).\]
Equation~\eref{emp} can also be represented in the form
\[ \rme^{-\varrho} = \alpha_{i+1} \rme^v \left(D_x + \beta_{i+1} \rme^{-v} - (\ln\alpha_{i+1})_x \right) (\rme^{-v}). \] 
Substituting $\varrho$ for $u$ into the first formula of~\eref{abei} and using the preceding equalities, we obtain by Theorem~\ref{ljmt} that the minimal $x$-integral of~\eref{abenp} is 
\[ D_x\left(\ln\Upsilon_i\left(\alpha_{i+1} \rme^v \left(D_x + \beta_{i+1} \rme^{-v} - (\ln\alpha_{i+1})_x \right) (\rme^{-v})\right)\right) + \beta_i \rme^{-\varrho} \] 
\[ = D_x(\ln(\alpha_{i+1} \rme^{v} \Upsilon_{i+1}(\rme^{-v}))) + \beta_{i+1} \rme^{-v} - (\ln\alpha_{i+1})_x - v_x \]
\[ = D_x(\ln\Upsilon_{i+1}(\rme^{-v})) + \beta_{i+1} \rme^{-v}. \]
This completes the proof of Theorem~\ref{main} by induction on i.
\end{proof}

According to~\cite{ZhIzv}, minimal integrals of orders $k>1$ can be chosen so that they depend linearly on the $k$-th derivative of $u$. In addition, it is not difficult to check that the orders of the leading coefficients in such integrals are uniquely defined (see subsection~3.2 in~\cite{Sljm} for details) and are preserved under point transformations of~\eref{hyp}. The leading coefficients $-\rme^{-u}/\Upsilon_i(\rme^{-u})$ and $\rme^u /\Lambda_i(\rme^{u})$ of~\eref{abei} respectively have the orders $i$ and $n-i$. This means that  the equations~\eref{abe} are very likely new in the generic case (at least under the condition $\max(i,n-i) > 2$) because the orders of the leading coefficients in the minimal integrals are less than 3 for all other Darboux integrable equations~\eref{hyp} known to the author. Since the equations~\eref{abe} admitting first-order integrals (i.e., the cases $i=0$ and $i=n$) were obtained in the recent work~\cite{St25}, genuinely new Darboux integrable equations arise for $i$ ranging from $1$ to $n-1$.

For example, if $n=3$, then Theorem~\ref{main} for $i=1$ and $i=2$ gives the equations with minimal integrals of order 3 in one characteristic and of order 2 in the other. To the author's best knowledge, all other equations with minimal integrals of such orders  either depend nonlinearly on the derivatives (see the Laine equations in \cite{ZhYu} and class 3 in \cite{ZhSok}) or are fully linear (with zero orders of the leading coefficients in both the minimal integrals). Neither of the above cases can be related via point transformations to any equation~\eref{abe} because the latter depends linearly on the derivatives and has the leading coefficient of non-zero order in at least one of the minimal integrals. Therefore, the equations~\eref{abe} with $n=3$ and $i \in \{1,2\}$ are also new Darboux integrable equations, although the condition $\max(i,n-i) > 2$ fails to hold in this case.

\end{document}